\documentclass[10pt]{article}

\usepackage[T1]{fontenc}
\usepackage[utf8]{inputenc}
\usepackage{lmodern}
\usepackage[margin=1in]{geometry}
\usepackage{microtype}
\usepackage{amsmath,amssymb,amsthm}
\usepackage{booktabs}
\usepackage{graphicx}
\usepackage{float}
\usepackage{enumitem}
\usepackage[numbers,sort&compress]{natbib}
\usepackage{xcolor}
\usepackage{url}
\usepackage[hidelinks]{hyperref}
\hypersetup{
  pdftitle={AccretionLink: On-Device Auditing of Exposure-Control Attacks on Attribute Inference},
  pdfauthor={Faruk Alpay and Taylan Alpay}
}

\newif\ifresultsready
\resultsreadytrue
\newif\ifdeviceresultsready
\deviceresultsreadytrue

\newcommand{\IfDeviceResults}[2]{\ifdeviceresultsready#1\else#2\fi}

\newcommand{\AggregatePrefixEight}{0.01595}
\newcommand{\AggregatePrefixEightCI}{[0.00890,\,0.02336]}
\newcommand{\IntegrityWins}{6/109}

\newcommand{\ModelSha}{e49a10b576929695}
\newcommand{\NpuMedianMs}{51.042}
\newcommand{\NpuPnnMs}{52.077}
\newcommand{\TokenizerMedianMs}{2.284}
\newcommand{\TokenizerPnnMs}{4.988}
\newcommand{\StateHeadMedianMs}{0.0458}
\newcommand{\StateHeadRatio}{0.000891}

\title{AccretionLink: On-Device Auditing of Exposure-Control Attacks on Attribute Inference}

\author{
Faruk Alpay$^{1}$\thanks{Corresponding author: \texttt{alpay@lightcap.ai}.}
\and Taylan Alpay$^{2}$\\[0.35em]
\small $^{1}$Department of Computer Engineering, Bah\c{c}e\c{s}ehir University, Istanbul, T\"urkiye\\
\small $^{2}$Department of Aerospace, University of Turkish Aeronautical Association, Ankara, T\"urkiye\\[0.35em]
\small Faruk Alpay: \texttt{faruk.alpay@bahcesehir.edu.tr}\\
\small Taylan Alpay: \texttt{s220112602@stu.thk.edu.tr}
}

\date{}

\newtheorem{proposition}{Proposition}
\newtheorem{theorem}{Theorem}
\theoremstyle{definition}
\newtheorem{remark}{Remark}

\setlist[itemize]{leftmargin=*,topsep=3pt,itemsep=2pt}
\begin{document}
\maketitle

\begin{abstract}
Exposure control lets an adversary rank authentic public posts to strengthen
private-attribute inference without altering content. AccretionLink defines
confidentiality and integrity games for this attack, models bounded selection
odds through partial identification, and constructs dependence-aware
time-uniform e-processes. On 52 held-out synthetic profiles, odds-four
selection reduced aggregate negative log likelihood at every horizon. At
eight posts the advantage was \AggregatePrefixEight{} nats (95\% CI
\AggregatePrefixEightCI), three of four target effects survived Holm
adjustment, and label-blind model-guided selection caused \IntegrityWins{}
high-confidence false reversals. On 142 PAN15 test profiles, exploratory selection produced a
0.01227-nat advantage but no reversal. A separate TF--IDF selector retained a
0.01470-nat advantage against the unchanged G5 target, while matched identity
shuffling did not reproduce it. Pixel 10 encoded all 1,622 held-out posts once
with a fallback-free Tensor G5 graph. A P-256 checkpoint authenticated the
selected-replay, actual-model, native-report, and operation digests; local
\texttt{KeyInfo} identified the signing key as StrongBox-backed.

\end{abstract}

\section{Introduction}
\label{sec:introduction}

Public content is not necessarily harmless content. Attribute-inference
attacks combine apparently ordinary activity to recover information a user did
not state in the attacked context\citep{gong2016attribute,staab2024beyond}.
Content can remain authentic and access-authorized while the observation
transcript is adversarially chosen. A ranking service, recommender,
intermediary, or local collector can control which legitimate items reach a
profiler and in what order. This exposure-only capability changes the evidence
available to attribute inference without compromising an account, fabricating
a post, or altering a signed object.

We call this capability \emph{exposure control}. It creates two security
failures. The confidentiality failure is a reduction in the adversary's proper
log loss relative to a paired uniform exposure of the same size. The integrity
failure occurs when selection turns a profile that the same adversary
classified correctly under uniform exposure into a high-confidence false
profile. The latter matters whenever inferred profiles are consumed as stable
facts. Both attacks operate entirely on authentic source material.

Exposure control also makes the observed stream statistically ambiguous. A
high cue rate may reflect high source prevalence, aggressive selection, or
both. AccretionLink models the selector as a bounded odds channel and audits
the stream at two levels. Partial identification describes which source
prevalences remain compatible with the selected observations. A cue-rate
e-process raises a time-uniform warning when the declared selection bound is
violated. A separate sealed-label block e-process consumes the frozen attack
posterior and tests whether text and synthetic profile labels remain
exchangeable. Predictive loss, leakage bits, and block evidence consequently
refer to one adversary and one security experiment.

The audit executes where the candidate posts reside. Sending raw drafts to a
remote profiler would create another disclosure channel. Re-encoding every
prefix would also place repeated host and accelerator work on the security
path. AccretionLink maps each admitted post once with a Tensor-G5-AOT semantic
encoder, updates a fixed-size prefix state, and advances the posterior and
evidence recurrences online. Separate device traces record NPU-only execution,
zero fallback nodes, asynchronous fence return, once-only held-out encoding,
and the relative service times of tokenization, G5 dispatch, and posterior
update. The G5 model/runtime boundary dominates those measured stages; this
localizes the registered-stream bottleneck beyond the measured CPU host stages.

The contributions are:

\begin{itemize}
    \item A formal exposure-control adversary that changes only the authentic
    observation transcript, with proper-score confidentiality and
    high-confidence integrity games.
    \item A dependence-safe audit from release-odds bounds to partial
    identification and anytime-valid evidence, including a finite-horizon
    false-warning inflation bound for understated selection limits.
    \item A Pixel 10 realization with once-only, fallback-free Tensor G5
    extraction and a replay verifier that commits selected-input, actual-model,
    native-report, operation, and chain-state digests to a P-256 checkpoint.
\end{itemize}

On sealed SynthPAI profiles, three of four target effects survive Holm
correction and model-guided selection yields six high-confidence false
profiles. An exploratory PAN15 check retains the proper-score direction on
real profiles but yields no such reversal. A post-hoc sensitivity control uses
a separate lexical selector against the unchanged G5 target, so the observed
score advantage does not require the target model to select its own inputs.

\section{Related work}

\paragraph{Attribute inference as a security attack.}
Attribute inference uses public behavior to recover fields a user has not
disclosed to the attacker. Social-graph and behavioral attacks established
this as an Internet privacy threat before current language models
\citep{gong2016attribute}; modern text models recover attributes from sparse,
ordinary posts\citep{staab2024beyond}. AttriGuard perturbs public features to
reduce attribute inference\citep{jia2018attriguard}. AccretionLink studies a
different adversarial control surface. Its selector leaves every feature
unchanged and chooses only the authentic exposed subset and order. SynthPAI
provides synthetic forum profiles with controlled labels, allowing this attack
to be evaluated without constructing a new real-person corpus
\citep{yukhymenko2024synthpai}.

The PAN shared tasks establish age and other author-profiling benchmarks on
Twitter text\citep{rangel2015pan,pan15}. We use the English PAN15 archive only
as a separate exploratory stress test. It contributes real-profile text but no
observed ranking log, so it cannot identify an actual platform's selection
policy.

Cumulative profiling is established. PAPI and HolmesEye aggregate multiple
personal images\citep{liu2025holmes}; SopriBench and Argus study cross-post,
user-level multimodal leakage\citep{peng2026sopribench}. Our object is the
hidden exposure mechanism. Two observers can see
sequences of equal length but receive different privacy evidence because one
sequence was selected adaptively.

\paragraph{Selection and partial identification.}
Selection can make population quantities unidentified without assumptions on
the observation mechanism. Sensitivity analysis and partial identification
therefore report sets of compatible quantities rather than a point estimate
\citep{manski2003partial,rosenbaum2002observational}. The odds-multiplier model
isolates each conclusion's dependence on a user-supplied selection bound
without positing a known platform policy. The finite benchmark uses an
exponential-race form of
weighted sampling without replacement, related to Plackett-Luce ranking
models\citep{luce1959choice,plackett1975permutations}.

\paragraph{Anytime-valid evidence and predictive calibration.}
Nonnegative test martingales and e-processes permit time-uniform testing under
optional stopping\citep{ville1939etude,howard2020timeuniform,ramdas2023savi}.
Predictable betting gives constructive processes for bounded observations
\citep{waudbysmith2024betting}, while arithmetic averaging preserves e-value
validity without independence\citep{vovk2021evalues}. Building on these
results, AccretionLink derives predictable Bernoulli ceilings from
remaining-inventory prevalence and release-odds bounds, aggregates dependent
streams within profile blocks, and bounds false-warning inflation when the
declared selection limit is understated. A second block e-process converts
frozen classifier posteriors into leakage bits and anytime-valid evidence under
a finite-population label-permutation null. Classifier probabilities
are temperature scaled on a calibration split\citep{guo2017calibration}.
Expected log posterior gain can be interpreted through a variational mutual-
information lower bound under the stated population-prior conditions
\citep{barber2003im}; an individual realization has no such population
interpretation.

\paragraph{Local integrity records.}
Local computation reduces transmission of draft content, but local execution
alone does not establish correctness. The audit record uses SHA-256 and ECDSA
as standardized primitives\citep{nistsha,nistdss} and Android Keystore for key
isolation\citep{androidkeystore}. A signed hash-chain checkpoint authenticates
the exported prefix under the local key. It provides neither remote
attestation nor proof that the model or its outputs are correct.

\section{Exposure-control security game}
\label{sec:threat}

\subsection{Actors and capabilities}

For synthetic profile $u$, let $Y_u^{(a)}\in\mathcal Y_a$ be the sealed value
of target $a$ and let
$\mathcal O_u=\{O_{u,1},\ldots,O_{u,n_u}\}$ be its authentic post pool. The
challenger gives $\mathcal O_u$, a frozen profiler $Q$, and a budget $h$ to a
selector $\mathcal A_S$. The selector returns distinct indices
$S_{u,1:h}$ and may adapt each choice to earlier profiler outputs. It never
receives $Y_u^{(a)}$. The profiler maps the selected prefix to a normalized
posterior
\begin{equation}
q_{u,h}^{(a)}=Q_a(O_{u,S_{u,1}},\ldots,O_{u,S_{u,h}}).
\label{eq:attackposterior}
\end{equation}
Selector and profiler may collude and know all model parameters, class priors,
candidate posts, and past outputs. The controlled attack additionally receives
the predeclared cue bit for each post. The model-guided attack receives no cue
metadata and may query $Q$ on remaining authentic candidates.

The attack has no account credential, network interception capability, or
write access to post contents. It cannot insert, edit, paraphrase, or duplicate
an item, and it cannot inspect a held-out label. Each admitted post retains its
source digest. Exact repetition is a dependence diagnostic outside the primary
game. This boundary isolates exposure control from content poisoning and
account compromise.

\subsection{Confidentiality and integrity games}

The challenger uses common random keys to produce paired uniform and selected
sequences of equal length. Write $q_{u,h,1}^{(a)}$ for uniform exposure and
$q_{u,h,4}^{(a)}$ for the registered $\lambda=4$ exposure. With log loss
$\ell(q,y)=-\log q(y)$, the empirical confidentiality advantage is
\begin{equation}
\operatorname{Adv}^{(a)}_{\mathrm{conf}}(h)
=\mathbb E_u\!\left[
\ell(q_{u,h,1}^{(a)},Y_u^{(a)})
-\ell(q_{u,h,4}^{(a)},Y_u^{(a)})
\right].
\label{eq:confadv}
\end{equation}
A positive value means that exposure control reduced the fitted adversary's
loss without increasing its observation budget. This is a proper-score attack
advantage over a paired release policy, rather than a claim of cryptographic
indistinguishability.

The integrity game asks whether selection can manufacture a trusted but false
profile from unchanged evidence. At confidence threshold $\gamma=0.8$, profile
$u$ is an integrity win for target $a$ when
\begin{equation}
\begin{split}
W_{u}^{(a)}={}&
\mathbf 1\{\arg\max_y q_{u,8,1}^{(a)}(y)=Y_u^{(a)}\}\,
\mathbf 1\{\arg\max_y q_{u,8,4}^{(a)}(y)\neq Y_u^{(a)}\}\,\\
&\mathbf 1\{\max_y q_{u,8,4}^{(a)}(y)\geq\gamma\}.
\end{split}
\label{eq:intwin}
\end{equation}
We report $\Pr(W_u^{(a)}=1)$ among profiles classified correctly under the
paired uniform sequence. The selector remains label-blind throughout the
attack; labels are unsealed only for scoring.

These games cover two exposure adversaries. A bounded selector changes release
odds for a predeclared cue by $\lambda\in\{2,4\}$. A score-guided selector
chooses the remaining post that maximizes its frozen objective after the
candidate prefix update. The first identifies a controlled selection axis; the
second realizes a model-aware integrity attack.

\subsection{Local auditor and evidence boundary}

The defender is a local auditor with legitimate access to the candidate pool.
It executes the same frozen profiler under paired release policies, estimates
Equations~\eqref{eq:confadv} and \eqref{eq:intwin}, and advances the sequential
tests in Section~\ref{sec:method}. The auditor detects exposure-conditioned
leakage; it does not sanitize a post or prevent a remote selector from acting.
Its Android package requests no Internet permission, and raw text is excluded
from the exported transcript. A compromised operating system, keyboard,
display capture path, or signing key lies outside the device threat boundary.

Each record commits content, model, execution-evidence, and output digest
fields to the previous chain root. The record layer has a separate tampering game. Let $H$
be SHA-256, $\mathsf{Sig}$ be ECDSA P-256 with SHA-256, and $pk^\star$ be the
audit public key fixed by a verifier before it receives a disputed export. For
record $R_i=(i,t_i,e_i,d_i^{\rm in},d_i^{\rm model},d_i^{\rm exec},
d_i^{\rm out},s_i)$, the implementation computes
\begin{equation}
  r_i=H(\mathsf{enc}(R_i,r_{i-1})),\qquad r_{-1}=0^{256},
  \label{eq:recordchain}
\end{equation}
where $\mathsf{enc}$ is the length-prefixed, domain-separated byte encoding in
the Android source. Strings and byte arrays are preceded by their 32-bit
big-endian length; the index and IEEE-754 score bits use 64-bit big-endian
words; the record domain is \texttt{accretionlink-audit-v2}. A checkpoint is
$C=(n,t_C,r_{n-1},\sigma)$ with
\begin{equation}
  \sigma\leftarrow
  \mathsf{Sign}_{sk^\star}(\mathsf{enc}_{C}(n,t_C,r_{n-1})).
  \label{eq:checkpoint}
\end{equation}
The checkpoint domain is
\texttt{accretionlink-audit-v2-checkpoint}.
Here $T$ denotes the projection onto the canonical fields of $R_i$.
Additional JSON members lie outside the authenticated transcript and must not
be interpreted as signed claims.

The explicit verifier $V(pk^\star,T,C,S)$ accepts only if indices in transcript
$T$ are exactly $0,\ldots,n-1$, the first predecessor is $0^{256}$, every later
predecessor equals the preceding recomputed hash, every digest is 32 bytes,
$C$ carries the same $n$ and terminal root, and the ECDSA signature verifies
under $pk^\star$. If prior state $S=(n_0,r_0)$ has been fixed outside the
device, $V$ additionally requires $n\geq n_0$ and that the prefix of $n_0$
records has root $r_0$. The tampering adversary obtains checkpoints on chosen
valid prefixes and wins by producing an accepted transcript that diverges from
the corresponding signed history, or rolls back past $S$.

Android \texttt{KeyInfo} locally reported the measured signing key as
StrongBox-backed\citep{androidkeyinfo}. The signed checkpoint authenticates the
record count, terminal root, and checkpoint time under the verifier-pinned
P-256 key. Its exported \texttt{key\_security\_level} label is not part of the
signed message, and the export contains no challenge-bound Android
key-attestation certificate chain. It therefore provides no authenticated
StrongBox provenance to a remote verifier\citep{androidkeyattestation}.
Appendix~\ref{app:transcript-proof} gives the integrity reduction and rollback
boundary.

The registered benchmark also supplies a sealed-label block audit. Text and
the complete posterior vector for profile $u$ are fixed before its label tuple
is revealed. Under the finite-population permutation null, label tuples are
assigned without replacement to the frozen profile sequences. Immediately
before block $u$, the predictable design prior for target $a$ is
\begin{equation}
\pi_{u}^{(a)}(y)=
\frac{\text{remaining labels of class }y}
     {\text{remaining profile blocks}}.
\label{eq:remainingprior}
\end{equation}
The canonical profile order, horizon mixture, and betting fractions are fixed
before test access. The resulting process tests label--text exchangeability
while respecting within-profile dependence.

\section{Selection-conditioned audit}
\label{sec:method}

\subsection{Selection channel and non-identifiability}

For one target, let $C\in\{0,1\}$ denote a predeclared observable cue and let
$\theta=\Pr(C=1)$ in the candidate collection. Suppose a cue-bearing item has
release weight $\lambda>0$ relative to a non-cue item. For a weighted draw, the
released cue probability is
\begin{equation}
r=\frac{\lambda\theta}{1-\theta+\lambda\theta},
\qquad
\operatorname{odds}(r)=\lambda\operatorname{odds}(\theta).
\label{eq:selection}
\end{equation}

\begin{proposition}[Partial identification under bounded selection]
If $0<L\leq\lambda\leq U$ and the selected-stream cue rate $r$ is known, then
the compatible source prevalence lies in
\begin{equation}
\theta\in
\left[
\frac{r}{U+(1-U)r},
\frac{r}{L+(1-L)r}
\right].
\label{eq:interval}
\end{equation}
Without a restriction on $\lambda$, $r\in(0,1)$ does not identify $\theta$.
\end{proposition}

\begin{proof}
Solving Equation~\eqref{eq:selection} gives
$\theta=r/[\lambda+(1-\lambda)r]$. This expression decreases in $\lambda$ for
$r\in(0,1)$, so its extrema occur at $U$ and $L$. For any
$\theta\in(0,1)$, choosing $\lambda=\operatorname{odds}(r)/
\operatorname{odds}(\theta)$ produces the same $r$; hence the unrestricted
model is not identified.
\end{proof}

The symmetric sensitivity specification used in the experiments is
$[L,U]=[1/U,U]$. Equation~\eqref{eq:selection} is a one-step sensitivity model,
not an assertion that finite feeds are sampled with replacement. In the
controlled benchmark, weighted exponential-race keys generate nested prefixes
without replacement. As the inventory changes, Equation~\eqref{eq:selection}
applies conditionally only when $\theta$ is interpreted for the current risk
set.

\subsection{What posterior NLL measures}

Fix a release policy $\lambda$ and write
$X_{1:t}^{\lambda}$ for its selected prefix.  Let $P_{\lambda}$ be the joint
law of $(Y,X_{1:t}^{\lambda})$, let
$\pi_{\lambda}(y)=P_{\lambda}(Y=y)$, and require the frozen adversary
$q_{\lambda,t}(\cdot\mid X_{1:t}^{\lambda})$ to be strictly positive.

\begin{proposition}[Attacker-achievable exposure identity]
With logarithms in base two,
\begin{align}
\mathcal{L}_{\lambda,t}
&:=\mathbb{E}_{P_{\lambda}}
  \log_2\frac{q_{\lambda,t}(Y\mid X_{1:t}^{\lambda})}
                   {\pi_{\lambda}(Y)} \nonumber\\
&=I_{2,P_{\lambda}}(Y;X_{1:t}^{\lambda})
-\mathbb{E}_{X_{1:t}^{\lambda}}
 D_{\mathrm{KL},2}\!\left(
 P_{\lambda}(Y\mid X_{1:t}^{\lambda})\,
 \middle\|\,
 q_{\lambda,t}(\cdot\mid X_{1:t}^{\lambda})
 \right)
\leq I_{2,P_{\lambda}}(Y;X_{1:t}^{\lambda}).
\label{eq:exposureidentity}
\end{align}
If a reference prior $\pi_0$ replaces the population marginal, the expected
log-score gain acquires the additive term
$D_{\mathrm{KL},2}(P_{\lambda}(Y)\|\pi_0)$ and need not be a mutual-
information lower bound.
\end{proposition}

\begin{proof}
Add and subtract
$\log_2 P_{\lambda}(Y\mid X_{1:t}^{\lambda})$ inside the expectation.
The first resulting expectation is mutual information, and the second is the
negative posterior KL regret. Replacing $\pi_{\lambda}$ by $\pi_0$ additionally adds
$\mathbb{E}\log_2[\pi_{\lambda}(Y)/\pi_0(Y)]$.
\end{proof}

Equation~\eqref{eq:exposureidentity} decomposes achievable log-score gain into
mutual information minus posterior regret, making NLL the operational
predictive quantity and exposing where calibration error enters. The paired
release regimes use the same profiles, labels, and Monte Carlo orders, so their
empirical label entropy cancels exactly:
\begin{equation}
\widehat{\mathrm{NLL}}_{\lambda=1,t}
-\widehat{\mathrm{NLL}}_{\lambda=4,t}
=\widehat{\mathcal L}_{\lambda=4,t}
-\widehat{\mathcal L}_{\lambda=1,t}.
\label{eq:nllcontrast}
\end{equation}
The contrast therefore measures increased exposure for the fitted adversary.
Because posterior regret may differ between release regimes, it need not equal
the change in true mutual information. We separately report the training-prior
log-score gain and the held-out empirical-entropy plug-in version.

\subsection{Selection-adjusted anytime warning}

The operational null is heterogeneous. Index synthetic profiles by $u$, and
predeclared target/order streams within a profile by $k$.  Before cue
$C_{u,k,t}\in\{0,1\}$ is revealed, the audit supplies a predictable bound
$\tau_{u,k,t}$ on the cue prevalence in the remaining source risk set and a
predictable upper bound $U_{u,k,t}$ on the release odds multiplier.  Write
\begin{equation}
\phi(v,x)=\frac{vx}{1-x+vx},\qquad
p^*_{u,k,t}=\phi(U_{u,k,t},\tau_{u,k,t}).
\label{eq:p0}
\end{equation}
The composite conditional null is the collection of all data laws satisfying
\begin{equation}
H_0^{\mathrm{sel}}:\quad
p_{u,k,t}:=\Pr(C_{u,k,t}=1\mid\mathcal F_{u,k,t-1})
\leq p^*_{u,k,t}\quad\text{for every }(u,k,t).
\label{eq:conditionalnull}
\end{equation}
Here $\mathcal F_{u,k,t-1}$ contains the earlier profile blocks, the released
past of this stream, and all predictable audit choices.  Equation~\eqref{eq:selection}
and monotonicity of $\phi$ show that the source and multiplier bounds imply
Equation~\eqref{eq:conditionalnull}.  The bounds may vary by profile and risk
set; the null imposes no independence or stationarity.

For an odds-inflation alternative $\delta>1$, define the closed-form process
\begin{equation}
E_{u,k,t}(\delta)=\prod_{s=1}^{t}
\frac{\delta^{C_{u,k,s}}}
     {1+(\delta-1)p^*_{u,k,s}},
\qquad E_{u,k,0}(\delta)=1.
\label{eq:eprocess}
\end{equation}
At a constant bound $p_0$, this is the Bernoulli likelihood ratio with
$q=\phi(\delta,p_0)$:
$(q/p_0)^C\{(1-q)/(1-p_0)\}^{1-C}$. Thus a fixed grid of interpretable odds
departures yields the finite mixture
$E^{\mathrm{mix}}_{u,k,t}=\sum_jw_jE_{u,k,t}(\delta_j)$, where
$w_j\geq0$ and $\sum_jw_j=1$ are frozen before the stream is observed. Below,
$E_t(q)$ denotes this equivalent constant-bound parametrization.

\begin{theorem}[Heterogeneous selection evidence and dependent-profile blocks]
\label{thm:heterogeneous-selection}
Under $H_0^{\mathrm{sel}}$, every $E_{u,k,t}(\delta)$ and every fixed mixture
$E^{\mathrm{mix}}_{u,k,t}$ is a nonnegative supermartingale in its stream, so
\begin{equation}
\Pr_{H_0^{\mathrm{sel}}}\!\left(
\sup_{t\geq0}E^{\mathrm{mix}}_{u,k,t}\geq\alpha^{-1}
\right)\leq\alpha.
\label{eq:ville}
\end{equation}
Let $e_{u,k}$ be a terminal stream e-value and let weights $a_{u,k}\geq0$,
$\sum_ka_{u,k}=1$, be chosen before the current profile's cues are observed.
Then $B_u=\sum_ka_{u,k}e_{u,k}$ remains an e-value conditional on earlier
profiles even when the streams within profile $u$ are arbitrarily dependent.
For any earlier-profile-measurable $\eta_u\in[0,1]$,
\[
M_n^{\mathrm{sel}}=\prod_{u=1}^{n}\{1-\eta_u+\eta_uB_u\}
\]
is a nonnegative supermartingale and obeys
$\Pr(\sup_nM_n^{\mathrm{sel}}\geq1/\alpha)\leq\alpha$.
\end{theorem}

Appendix~\ref{app:selection-proofs} proves the result through one-step
conditional supermartingale arguments and arithmetic e-value closure
\citep{ville1939etude,vovk2021evalues,ramdas2023savi}.
It also permits a predictable betting alternative
\[
E^{\mathrm{bet}}_{u,k,t}
=\prod_{s=1}^{t}\{1+\nu_{u,k,s}(C_{u,k,s}-p^*_{u,k,s})\},
\qquad 0\leq\nu_{u,k,s}\leq1/p^*_{u,k,s},
\]
where each stake is chosen before $C_{u,k,s}$ is observed
\citep{howard2020timeuniform,waudbysmith2024betting}. Streams that reuse posts
or orders are mixed within a profile, never multiplied; only profile blocks
advance $M_n^{\mathrm{sel}}$.

\begin{remark}[What the guarantee requires]
The theorem is invalid if the actual conditional selection multiplier exceeds
$U$, if $\tau$ fails for a remaining risk set, or if the cue is defined after
examining the same sequence. Exact repeats generally violate a naive
conditional model. The implementation therefore deduplicates exact text before
applying the warning and treats near-duplicate clustering as a required
sensitivity analysis. $E_t$ is evidence against a cue-rate null; it is not
$q_t(Y)$ and cannot be read as a persona probability.
\end{remark}

\begin{proposition}[Inflation when the selection bound is understated]
\label{prop:misspecified-u}
Suppose the audit uses $\widetilde p_i=\phi(\widetilde U_i,\tau_i)$ but the true
null only guarantees $p_i\leq p_i^+=\phi(U_i^+,\tau_i)$, with
$U_i^+\geq\widetilde U_i$. For mixture component $j$, set
\[
\kappa_{i,j}=\frac{1+(\delta_j-1)p_i^+}
                    {1+(\delta_j-1)\widetilde p_i},\qquad
K_T^{\max}=\prod_{i=1}^{T}\max_j\kappa_{i,j}.
\]
Then the warning computed from the understated bounds satisfies the finite-
horizon envelope
\begin{equation}
\Pr\!\left(\sup_{t\leq T}E_t^{\mathrm{mix}}\geq\alpha^{-1}\right)
\leq\min\{1,\alpha K_T^{\max}\}.
\label{eq:misspecinflation}
\end{equation}
For constant assumed and true boundary rates $a<b$, the worst-case i.i.d.
log-growth of component $\delta$ is
\begin{equation}
b\log\delta-\log\{1+(\delta-1)a\}
=D_{\mathrm{KL}}(b\|a)-D_{\mathrm{KL}}(b\|\phi(\delta,a)).
\label{eq:misspecgrowth}
\end{equation}
It is maximized at $\delta=\operatorname{odds}(b)/\operatorname{odds}(a)$,
where the drift is $D_{\mathrm{KL}}(b\|a)>0$.
\end{proposition}

Thus underestimating $U$ can create exponential, rather than merely additive,
false-warning inflation. Equation~\eqref{eq:misspecinflation} is a sensitivity
envelope, not a recovered level-$\alpha$ guarantee; the primary guarantee still
requires the declared conditional bound to hold.

\begin{proposition}[Finite-horizon crossing diagnostic]
\label{prop:finite-crossing}
For the diagnostic model $C_i\overset{\mathrm{iid}}{\sim}
\operatorname{Bernoulli}(r)$, a constant $p_0<r$, and fixed $q>p_0$, let
\[
A=\log(q/p_0),\quad B=\log\{(1-q)/(1-p_0)\},\quad D=A-B,
\]
and $h_T=\left\lceil\{\log(1/\alpha)-TB\}/D\right\rceil$.  With
$\tau_\alpha=\inf\{t:E_t(q)\geq1/\alpha\}$,
\begin{equation}
\Pr_r(\tau_\alpha\leq T)
\geq \Pr_r\{\operatorname{Bin}(T,r)\geq h_T\}
=\sum_{s=h_T}^{T}\binom{T}{s}r^s(1-r)^{T-s},
\label{eq:finitepower}
\end{equation}
with the usual zero/one boundary conventions. If
$x_T=\{\log(1/\alpha)/T-B\}/D<r$, Hoeffding's inequality
\citep{hoeffding1963probability} makes this at least
$1-\exp\{-2T(r-x_T)^2\}$.
Moreover,
\begin{equation}
\frac{1}{t}\log E_t(q)\longrightarrow
D_{\mathrm{KL}}(r\|p_0)-D_{\mathrm{KL}}(r\|q)
\quad\text{almost surely}.
\label{eq:egrowth}
\end{equation}
Positive drift therefore implies eventual crossing almost surely; at $q=r$,
$\log(1/\alpha)/D_{\mathrm{KL}}(r\|p_0)$ remains the first-order scale.
For a finite mixture, replacing $\alpha$ by $\alpha w_j$ in $h_T$ gives a
valid lower bound from any component $j$ because
$E_t^{\mathrm{mix}}\geq w_jE_t(q_j)$.
\end{proposition}

Appendix~\ref{app:selection-proofs} proves both propositions. These stationary
calculations are prospective diagnostics only. The anytime theorem itself
allows arbitrary predictable dependence satisfying Equation~\eqref{eq:conditionalnull}.

\subsection{The G5 posterior as block evidence}

The cue process above audits the release mechanism. A second, complementary
construction asks whether the G5 adversary extracts attribute information at
all. It uses the same posterior that defines NLL, rather than an unrelated
confidence heuristic. Index independent synthetic-profile audit blocks by
$u$. Let $Y_u\in\mathcal{Y}_a$ be the sealed target label, let $X_{u,1:h}$ be
the selected prefix, and let $\pi_u(y)>0$ be the frozen design prior. Before
unsealing $Y_u$, the Pixel computes a normalized vector
$q_{u,h}(\cdot\mid X_{u,1:h})$ for each registered horizon.

\begin{theorem}[Posterior block e-process]
Suppose the operational no-leakage null is
\begin{equation}
H_0^{\mathrm{attr}}:\quad
\Pr(Y_u=y\mid\mathcal F_{u-1},X_{u,1:H},C_u)=\pi_u(y),
\label{eq:attributenull}
\end{equation}
where $C_u$ contains public design context and the current label remains sealed
while the posterior and all horizon weights are chosen. Then, for every fixed
$h$,
\begin{equation}
e_{u,h}=\frac{q_{u,h}(Y_u\mid X_{u,1:h})}{\pi_u(Y_u)}
\label{eq:posteriorevalue}
\end{equation}
has conditional mean one under $H_0^{\mathrm{attr}}$. For predictable weights
$w_{u,h}\geq0$ summing to one,
$\bar e_u=\sum_h w_{u,h}e_{u,h}$ is also an e-value. If
$\eta_u\in[0,1]$ is chosen from past blocks only, then
\begin{equation}
M_n=\prod_{u=1}^{n}\{1-\eta_u+\eta_u\bar e_u\}
\label{eq:blockeprocess}
\end{equation}
is a nonnegative martingale under the null and obeys
$\Pr(\sup_nM_n\geq1/\alpha)\leq\alpha$.
\end{theorem}

\begin{proof}
Condition on the pre-unsealing sigma-field. Normalization gives
$\mathbb E_0[e_{u,h}\mid\cdot]
=\sum_y\pi_u(y)q_{u,h}(y)/\pi_u(y)=1$. Linearity proves the horizon mixture
claim without any independence among prefixes. The predictable betting factor
in Equation~\eqref{eq:blockeprocess} also has conditional mean one; iterated
conditioning and Ville's inequality complete the proof.
\end{proof}

Validity requires $q$, temperatures, horizon weights, and bets to be frozen
before the current label is unsealed. Calibration affects the log-score
interpretation but not e-validity. Repeated horizons sharing one sealed label
are arithmetically merged within a profile block rather than multiplied. Only
profile blocks advance $M_n$. A label-blind stopping rule based on entropy,
margin, or a resource
budget can replace a fixed horizon by treating the stopped prefix as the
block observation; stopping after viewing the true-class score is forbidden.

Equations~\eqref{eq:exposureidentity} and \eqref{eq:posteriorevalue} expose the
method--systems link. The Tensor-G5 posterior yields both attacker-achievable
leakage bits, $\log_2 e_{u,h}$, and anytime-valid block evidence. The CPU does
only the 18-logit normalization, horizon merge, and scalar betting recurrence;
semantic evidence is produced by the once-only G5 encoder stream.

\subsection{A sequential model that determines the execution graph}

The audit is defined on released prefixes, but an implementation need not run
an encoder on every prefix. Let $x_t$ be the normalized text of the newly
admitted post and let
\begin{equation}
e_t=f_{\phi}(x_t)\in\mathbb{R}^{d},\qquad
s_t=s_{t-1}+e_t,\qquad
z_t=\frac{s_t}{\max\{\lVert s_t\rVert_2,\varepsilon\}},
\label{eq:streamstate}
\end{equation}
with $s_0=0$. The encoder parameters $\phi$ are frozen before the held-out split
is opened. Each admitted post is embedded at most once and bound to its content
digest; later prefixes update only $s_t$. Thus $T$ distinct posts require
exactly $T$ encoder invocations and $O(Td)$ state-update work, rather than
reconstructing and re-encoding the $T$ overlapping prefixes. Exact repeats are
deduplicated for the warning as required above and cannot silently create an
extra encoder invocation.

The registered primary candidate uses the classification prompt and the
released SentencePiece vocabulary to produce a padded
\texttt{int32[1,256]} token vector. The frozen encoder is the 300M-parameter,
sequence-length-256 EmbeddingGemma deployment variant, which returns an
$\ell_2$-normalized \texttt{float32[1,768]} representation
\citep{embeddinggemma2025,embeddinggemmacard}. Its mixed-precision contract
uses per-channel INT4 quantization for embeddings, feed-forward layers, and
projections and INT8 quantization for attention. On Tensor G5 the registered
FlatBuffer is an ahead-of-time compiled artifact whose model graph consists of
a single custom \texttt{DISPATCH\_OP}; Section~\ref{sec:device} states the
evidence and the limits of that fact.

For each attribute $a$, a small transparent head computes
\begin{equation}
\ell_t^{(a)}=W_a z_t+b_a,
\qquad
q_t^{(a)}=\operatorname{softmax}\!\left(\ell_t^{(a)}/T_a\right).
\label{eq:heads}
\end{equation}
The four heads contain 18 logits in total. They are fitted on training profiles
with equal profile weight and class balancing; a deterministic internal split
of training profiles controls fitting choices. Calibration profiles are
reserved for the four temperatures $T_a$ and for the predeclared candidate
gate. No encoder or head is selected using held-out text, embeddings, labels,
predictions, or metrics. Posterior evaluation, Equation~\eqref{eq:eprocess},
Equation~\eqref{eq:streamstate}, and integrity bookkeeping are small host
operations and do not change the e-process guarantee.

The host reference already implements Equation~\eqref{eq:eprocess} as a
log-likelihood-ratio scalar recurrence with $O(1)$ work and state per observed
cue; its online updates are tested against the corresponding batch recurrence.
The registered host benchmark comprises the 768-dimensional running-sum
update, the $768\times18$ linear head, four softmax operations, and four
e-process updates. Constant prefix cost in this benchmark is necessary but not
sufficient to call the host work off-path: the device trace must also show that
submission of post $t+1$ does not wait for posterior or e-process work at
post $t$.

\paragraph{Hashed baseline and candidate gate.}
Signed hashed word unigrams, word bigrams, and character 3--5-grams in 4,096
dimensions followed by a two-layer 256/96 ReLU MLP provide the compact
baseline\citep{weinberger2009hashing}. Its generic TFLite graph contains
built-in operators but no Tensor G5 \texttt{DISPATCH\_OP}; its measured
execution is therefore a CPU baseline. Before test access, the semantic
candidate passed the predeclared calibration comparison against this baseline
and was frozen as the predictor. The selection estimands and e-processes are
unchanged by that model choice.

\section{Experimental protocol}
\label{sec:protocol}

\subsection{Data and partitioning}

SynthPAI contains 7,823 English comments from 300 fully synthetic profiles
\citep{yukhymenko2024synthpai}. We use the pinned dataset revision and verify
the acquired JSONL file by SHA-256. Raw data are not included in the paper
source or ancillary archive. Only comment text reaches the predictor. Profile
records supply four supervised targets; username, generated writing style,
author identifier, model guesses, full review content, and thread structure
are prohibited inputs.

Profiles, rather than comments or order replicates, are the independent units.
Selection experiments first retain the 261 profiles having at least eight
distinct comments. A deterministic multi-task stratification then assigns the
eligible profiles to 157 training, 52 calibration, and 52 test profiles. The
test split is opened once. The eligibility decision and split identifiers are
frozen before model fitting; no test text, feature, label distribution,
prediction, or metric is used to select the model.

The classes are: income \{low, middle, high\}, where the source labels
\emph{high} and \emph{very high} are merged; relationship
\{single, in a relationship, engaged, married, divorced, widowed\}; education
\{below bachelor's, bachelor's, graduate or professional, doctorate\}; and age
\{18--24, 25--34, 35--44, 45--54, 55 or older\}. These class definitions
follow the synthetic benchmark taxonomy.

\subsection{Release regimes}

For each target and eligible profile, the same deterministic random keys
define nested prefixes of lengths $1,2,4,8$. Seven order replicates are used on
the calibration split; the held-out evaluation uses 100 paired order
replicates. The regimes are:

\begin{enumerate}[leftmargin=*]
    \item \textbf{Uniform}: sampling without replacement, implemented as the
    exact $\lambda=1$ control for every profile and replicate.
    \item \textbf{Controlled odds intervention}: a target-specific binary
    stratum derived before fitting from parse-valid benchmark review metadata
    sets the exponential-race release rate to $\lambda\in\{2,4\}$; all other
    comments receive rate one. Neither the stratum nor review content enters
    the attribute model.
    \item \textbf{Confidence-greedy}: at each step, the selector chooses the
    remaining authentic comment that maximizes the current model's largest
    class probability after appending it. The selector receives neither the
    true label nor human-review metadata.
    \item \textbf{Exact-repeat diagnostic}: repeated copies of one comment are
    compared with exact-text deduplication. This is a negative control for
    dependence and is excluded from the primary attack estimand.
\end{enumerate}

The controlled stratum supplies a known release-odds axis;
confidence-greedy selection supplies the model-aware attack.

\subsection{Estimands and uncertainty}

The pre-specified primary estimand for each target is the profile-paired change
in negative log likelihood at prefix eight between $\lambda=4$ and
$\lambda=1$. Positive reported amplification means that selected release lowers
the adversary's log loss. Secondary metrics are posterior gain over the
training-profile prior, macro F1, balanced accuracy, multi-class Brier score,
and 10-bin expected calibration error. Baselines include the training-profile
prior, hard modal-frequency imputation with finite probability smoothing, and
the single-comment prefix.

For profile $u$, target $a$, and release regime $\lambda$, the 100 order
replicates are first averaged inside the profile. Writing
\begin{equation}
d_u^{(a)}=\frac{1}{R}\sum_{r=1}^{R}
\left[-\log q^{(a)}_{1,u,r}(Y_u)
      +\log q^{(a)}_{4,u,r}(Y_u)\right],
\qquad
\widehat\Delta_a=\frac{1}{N}\sum_{u=1}^{N}d_u^{(a)},
\label{eq:profileestimand}
\end{equation}
the inferential sample size is $N=52$, not $NR$. The Monte Carlo orders
reduce integration error for each $d_u^{(a)}$; they do not create new people
or independent observations. Equation~\eqref{eq:profileestimand} also makes
the label-entropy cancellation in Equation~\eqref{eq:nllcontrast} exact on the
paired empirical population.

Uncertainty intervals use 10,000 bootstrap replicates that resample complete
synthetic profiles and retain all paired orderings for a sampled profile
\citep{efron1993bootstrap}. The four target-wise primary tests are adjusted by
Holm's step-down procedure\citep{holm1979simple}. Order replicates increase
Monte Carlo stability but never increase the reported sample size. Exact
seeds, configurations, split digests, model hashes, and result JSON are part of
the ancillary record.

The mean-effect test uses profile-level paired differences only. Its
Rademacher sign-flip calibration is exact only under the stated
sign-exchangeability null for $d_u^{(a)}$; we therefore report the effect and
profile-bootstrap interval as primary evidence and the adjusted sign-flip
$p$-value as a supporting analysis, rather than calling it design-exact.
Before opening the test split, calibration profiles determine a nuisance-only
power curve: observed profile-level dispersion is held fixed, synthetic mean
shifts are added, and the complete bootstrap/sign-flip/Holm pipeline is
repeated. We report the 80\% detectable shift and power over a fixed grid. No
observed calibration effect is substituted for a held-out result.

For the anytime warning, matched uniform-release sequences estimate the
empirical familywise crossing rate. Fixed departures $r-p_0$ supply power and
median crossing-time diagnostics, which are compared with the KL growth in
Equation~\eqref{eq:egrowth}. Whole profiles, not posts, are resampled. A
misspecified-bound sweep $U_{\mathrm{assumed}}/U_{\mathrm{true}}$ exposes the
point at which nominal false-warning control is lost and is interpreted against
the finite-horizon envelope in Equation~\eqref{eq:misspecinflation}; it does not
estimate the platform's $U$.

\subsection{Analysis freeze}

Before test access, the experiment froze the data digest, split manifest,
feature contract, Tensor G5 FlatBuffer, linear heads, calibration temperatures,
selection policies, estimands, bootstrap seeds, and result generator. The
candidate manifest binds 29 files under SHA-256 digest
\texttt{1201b464a1c8bb92}. The held-out command opened the test split once after
that freeze and emitted a content-addressed result record. It processed 1,622
distinct test posts belonging to 52 profiles; candidate choice and calibration
were not revisited after the record was opened.

\subsection{PAN15 transfer protocol}

PAN15 was specified after completion of the primary synthetic analysis and is
therefore treated as exploratory. We evaluated the English portion of PAN15
Author Profiling\citep{pan15,rangel2015pan}.
The official archive contains Twitter text from 152 training profiles and 142
test profiles. Global exact-text deduplication removed 123 training and 207
test posts, leaving 14,043 and 12,971 posts respectively. The sole target is
the benchmark-provided four-way age group. The task documentation specifies
the classes but not the collection or label provenance; the analysis treats
them as benchmark labels rather than self-reports.

The external head is fitted only on the official training profiles. A seeded,
profile-stratified internal split reserves 122 profiles for fitting and 30 for
temperature calibration. Pixel 10 tokenizes and encodes each retained post once
with the same Tensor G5 artifact. The official test embeddings and labels are
opened only after the fitted head, temperature, source digests, seeds, and
evaluation program have been frozen.

This check replaces the synthetic review-derived cue with a label-blind,
model-derived score. For a single-post age posterior $q_i$, its salience is
$s_i=1-H(q_i)/\log 4$ and its release weight is $4^{s_i}\in[1,4]$.
Paired exponential-race orders compare this policy with uniform release at
prefixes $1,2,4,8$. It is a bounded selection simulation, not a reconstruction
of Twitter's ranking system. We report a profile-cluster bootstrap, a
profile-level sign-flip analysis, a permuted-label negative control, and the
same high-confidence reversal count. It is not pooled with the primary
synthetic hypothesis tests.

Because this salience is computed by the scored target model, we subsequently
specified a cross-model sensitivity control after the original PAN15 result
was known. The control is post-hoc and not preregistered. A multinomial
logistic selector uses word and character TF--IDF features fitted on the 122
official training profiles and temperature-scaled on the 30 calibration
profiles. Its single-post entropy defines the same odds-four release weights,
but the scored outcome remains the unchanged G5 embeddings and frozen age
head. Test labels never enter the selector. Hyperparameters, source digests,
and the fitted selector digest were fixed before the control test scores were
computed.

A matched negative control permutes the TF--IDF weights among posts within
each profile and paired draw. This preserves the exact weight distribution
while breaking the link between a post and its selection weight. The lexical
selector and G5 target have distinct representations and parameters, but they
share the age task and official training corpus. The control therefore tests
whether literal target-model reuse is necessary; it does not establish task
independence or identify a live platform policy.

\section{Held-out security evaluation}
\label{sec:results}

The frozen candidate was evaluated once on 52 previously sealed synthetic
profiles. Each profile contributed 100 paired release orders; these orders
integrate over the release mechanism and are not treated as additional people.
All intervals resample complete profiles and all four primary tests use the
pre-specified Holm correction.

\subsection{Confidentiality advantage}

Exposure control reduced the fitted adversary's aggregate NLL at every
registered horizon (Figure~\ref{fig:selection-exposure}). The paired
advantages for prefixes of one, two, four, and eight posts were respectively
0.03029 nats (95\% CI [0.01980, 0.04102]), 0.02800 [0.01760, 0.03889],
0.02337 [0.01397, 0.03287], and 0.01595 [0.00890, 0.02336]. The effect is
largest when the adversary sees little evidence and remains positive after
eight posts.

\begin{figure}[H]
\centering
\includegraphics[width=0.90\linewidth]{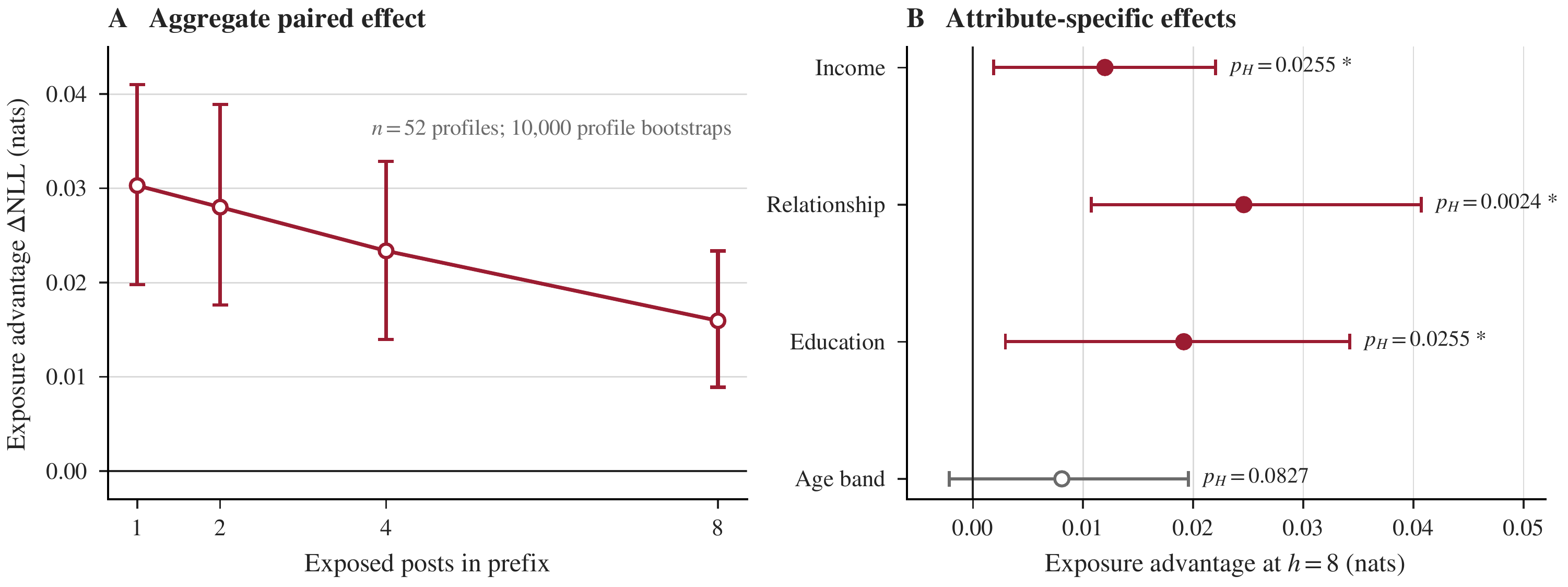}
\caption{Held-out confidentiality advantage under exposure control. Panel A
averages the four attributes at each prefix. Panel B gives the registered
attribute-specific effects at eight posts; filled markers denote Holm-adjusted
$p<0.05$. Error bars are 95\% profile-bootstrap intervals. The 100 paired
orders reduce Monte Carlo error within each of the 52 profiles and do not
increase the inferential sample size.}
\label{fig:selection-exposure}
\end{figure}

\noindent\begin{minipage}{\linewidth}
At the primary horizon of eight posts, income NLL changed from 0.7324 under
uniform release to 0.7204 under $\lambda=4$. Its paired advantage was 0.01199
nats (95\% CI [0.00188, 0.02204], Holm $p=0.0255$). Relationship NLL changed
from 1.6078 to 1.5832, an advantage of 0.02459 [0.01075, 0.04070] (Holm
$p=0.0024$). Education NLL changed from 0.9821 to 0.9630, an advantage of
0.01915 [0.00295, 0.03423] (Holm $p=0.0255$). Age-band NLL changed from 1.2169
to 1.2088, an advantage of 0.00808 whose interval included zero
([-0.00214, 0.01957], Holm $p=0.0827$). Exposure amplification therefore
survived familywise correction for three targets, with no corresponding
conclusion for age band.
\end{minipage}

\subsection{Integrity attack and sequential warning}

The model-guided selector generated high-confidence false profiles without
editing a post. It succeeded for 5 of 36 income cases eligible under the
paired uniform baseline and for 1 of 33 education cases. It succeeded in none
of 18 relationship or 22 age-band cases. The resulting 6 wins among 109
target-specific eligible cases are a descriptive security count because one
synthetic profile can enter more than one target denominator.

Under matched uniform release, the cue-rate e-process crossed its
$E_t\geq20$ boundary in 5 of 20,800 profile-order streams, a rate of 0.024\%.
The target counts were two for income, two for relationship, one for education,
and zero for age band. These streams comprise 100 Monte Carlo orders nested in
each profile; the observed rate is a negative-control diagnostic rather than a
binomial coverage estimate. Time-uniform type-I control follows from the
conditional supermartingale in Section~\ref{sec:method}.

\subsection{PAN15 transfer evaluation}

The PAN15 check was evaluated separately from the primary synthetic family. It
used 142 official English Twitter test profiles and 12,971 posts
after global exact-text deduplication. A Tensor G5 encoder processed each post
once with fallback disabled; the test extraction had a 51.085 ms median NPU
invocation time and a 2.198 ms median tokenizer time. The device remained in
the light thermal state throughout this long extraction, so these timings are
execution evidence rather than a cross-device performance comparison.

The label-blind salience policy improved the frozen age head's NLL at every
registered prefix. At one, two, four, and eight posts, the profile-averaged
uniform-minus-selected contrasts were 0.00868 [0.00553, 0.01209], 0.01029
[0.00686, 0.01382], 0.01195 [0.00865, 0.01544], and 0.01227 [0.00905,
0.01564] nats, respectively. The four exploratory sign-flip analyses each had
Holm-adjusted $p=0.00040$. At the primary external prefix, uniform NLL was
1.12209 and selected NLL was 1.10994. A 2,000-replicate permuted-label control
had mean 0.00029 and one-sided $p=0.00050$, whereas the observed contrast was
0.01227.

The post-hoc cross-model control addresses whether the original selector merely
favored inputs on which the scored model was already confident. The TF--IDF
selector and G5 target used distinct representations and fitted parameters.
Against the unchanged G5 target, its uniform-minus-selected contrasts were
positive at all four horizons; at eight posts the advantage was 0.01470 nats
(95\% CI [0.01188, 0.01764], Holm-adjusted $p=0.00040$). Permuting the same
selection weights among posts within each profile and draw did not reproduce a
positive advantage at eight posts ($-0.00088$ nats; Holm-adjusted $p=0.0676$).
The paired selector-minus-shuffle difference was 0.01558 nats [0.01279,
0.01856] (Holm-adjusted $p=0.00040$), and a 10,000-replicate profile-label
permutation gave one-sided $p=0.00010$. Post-level selector and target
saliences remained correlated (Spearman $\rho=0.255$), as expected for models
trained on the same age task.

These PAN15 results provide external support for the proper-score exposure
effect and show that literal selector--target reuse is not required for the
observed contrast. They do not support the integrity event: neither selector
caused a high-confidence false reversal among 71 eligible profiles. Because
the cross-model control was specified after the original PAN15 effect was
known, it is a sensitivity analysis rather than confirmatory evidence. PAN15
contains no observed exposure log and does not document how its 2015 benchmark
age groups were obtained; both policies characterize declared bounded
selection, not Twitter's ranking system.

\section{Pixel 10 security operation}
\label{sec:device}

The Android implementation links local semantic evidence extraction to a
registered exposure-control replay through content-addressed evidence records.
It requests no Internet permission. Raw text is consumed by the tokenizer and
encoder but is absent from the exported result and security replay.

\subsection{Tensor G5 execution contract}

The registered encoder is the sequence-length-256 mixed-precision
EmbeddingGemma-300M artifact compiled ahead of time for Tensor G5
\citep{tensorsdk,litert}. Its SHA-256 digest begins
\texttt{\ModelSha}. FlatBuffer inspection finds one subgraph containing one
custom \texttt{DISPATCH\_OP}. LiteRT 2.1.6 opens the model with
\texttt{Accelerator.NPU}, fallback disabled, and reports the compiled model as
fully accelerated with zero fallback nodes. The exercised tensor contract is
\texttt{int32[1,256]} to \texttt{float32[1,768]}. A generic mixed-precision
CPU reference and the G5 artifact produced cosine similarity 0.998933 on the
same token vector, with mean absolute error 0.001331 and maximum absolute error
0.005717.

The sealed held-out extraction began only after the candidate freeze. Android
SentencePiece tokenized 1,622 distinct posts from 52 test profiles and the G5
graph encoded the same 1,622 posts. Invocation counters establish one tokenizer
call and one NPU call per post, with no prefix-history re-encoding. Neither raw
text nor target labels occurs in the extracted metadata or embedding archive.
For every held-out horizon $H$, the extraction counters establish
\begin{equation}
N_{\mathrm{tokenize}}(H)=N_{\mathrm{enc}}(H)=H.
\label{eq:onceonly}
\end{equation}
The fixed-state recurrence was exercised separately on train/calibration
embeddings and verified one state update per admitted post. Horizons at one,
two, four, and eight posts therefore reuse accumulated embeddings rather than
invoking the semantic encoder again.

\subsection{Bottleneck localization}

Let $P$, $E$, and $C$ denote the per-post service times for SentencePiece,
Tensor G5 encoding, and the fixed state/head/e-process update. Once-only
encoding changes the steady-stream service interval from repeated prefix work
to
\begin{equation}
T_{\mathrm{service}}\simeq\max\{P,E,C\}
\label{eq:pipelinecriticalpath}
\end{equation}
when producer, accelerator, and consumer stages are connected by bounded
buffers. This decomposition makes host-stage decoupling a measurable security
systems property rather than an inference from the presence of an NPU label.

On the held-out raw-text stream, SentencePiece had median and 95th-percentile
times of \TokenizerMedianMs{} and \TokenizerPnnMs{} ms. The corresponding G5
invocation times were \NpuMedianMs{} and \NpuPnnMs{} ms. A separate ten-block
Pixel measurement over train/calibration embeddings placed the complete
768-dimensional state update, 18-logit heads, four softmax operations, and
four e-process recurrences at a median of \StateHeadMedianMs{} ms. Its median
ratio to the G5 stage was \StateHeadRatio{}, and its prefix-eight to prefix-one
95th-percentile ratio was 1.016. Thus the tokenizer is about 22 times faster
than the encoder and the posterior recurrence consumes below 0.1\% of the G5
service time. Under Equation~\eqref{eq:pipelinecriticalpath}, neither measured
host stage determines the stream rate; the dominant service boundary is the
Tensor G5 compiled model and its runtime.

The native \texttt{CompiledModel} test exercised two reusable tensor-buffer
sets through \texttt{RunAsync}\citep{litertcompiledmodel}. Both calls returned
\texttt{async=true} with unsignaled output events and valid fences, and the
second submission returned before the first event wait. The two-submit host
tail was 4.608 ms against a 51.176 ms synchronous single-call reference, a
ratio of 0.090. This establishes that host submission is decoupled from output
completion. The wall time through both fences was 0.936 times two synchronous
calls, above the registered 0.900 overlap gate. The trace establishes
asynchronous host release and localizes the dominant service stage. Because the
overlap gate was not met, overlap between G5 kernels remains unresolved.
These timing values come from the preregistered aggregate trace. Figure~\ref{fig:pixel-screen}
binds a distinct later live report, which produced the same no-fallback outcome
and likewise does not support a kernel-overlap claim.
The report is a local aggregate execution record rather than a freshness
attestation. It carries no verifier nonce or authenticated timestamp, and its
\texttt{heldout\_test\_accessed=false} field is interpreted inside the stated
application/OS trust boundary.

\subsection{Executable security game and integrity record}

The device operation consumes a registered audit bundle after verifying its
SHA-256 digest. The bundle contains hashed profile identifiers, profile-level
paired NLL effects, confidence-greedy integrity outcomes, and uniform-release
e-process records. It contains neither text nor embeddings. The Android engine
reaggregates target-wise profile means, prefix contrasts, and event counts and
rates from hashed profile or sequence records. It checks the supplied NLL
levels, confidence intervals, raw and adjusted $p$-values,
integrity-transition flags, and per-sequence warning flags for schema, range,
coverage, and stated internal relations. Labels, predictions, and per-step
e-values are absent, so those transitions and crossings are replay-validated
rather than regenerated. The plot is populated from this verified result,
not from registration metadata.

Each accepted operation appends a length-prefixed, domain-separated record.
For the executed operation, $d^{\rm in}$ is the digest of the selected replay
bytes, $d^{\rm model}$ is the digest of the actual on-device model bytes,
$d^{\rm exec}$ is the digest of the validated native G5 report bytes, and
$d^{\rm out}$ binds those three values to the replay-verification digest. The
engine rehashes the model and revalidates the native report immediately before
appending. This binding does not assert that replay aggregation ran on G5 or
that the replay effects were produced by that live invocation. The terminal
root and record count enter a
separately domain-separated ECDSA P-256 checkpoint signed by the Android
Keystore key. Section~\ref{sec:threat} defines the external verifier and
tampering game; Theorem~\ref{thm:transcript} reduces accepted divergence to
signature forgery, a SHA-256 collision, or violation of externally retained
rollback state. On the measured Pixel 10,
\texttt{KeyInfo.getSecurityLevel()} reported StrongBox. This is local
hardware-backed key evidence. Because the export contains no challenge-bound
attestation certificate chain and no off-device chain validation, it is not
remote device attestation. An earlier checkpoint retained independently is
required to detect prefix rollback.

Figure~\ref{fig:pixel-screen} records the completed operation after bundle
selection, digest verification, replay aggregation, runtime-evidence binding,
and chain append. It exposes the selected file, bundle digest, profile count,
operation identifier, verified outcomes, and two-record chain state. A
subsequent displayed-state commit produced the separately exported
three-record signed checkpoint.

\IfDeviceResults{
\IfFileExists{figures/pixel10_experiment.png}{
\begin{figure}[H]
\centering
\includegraphics[width=0.47\linewidth]{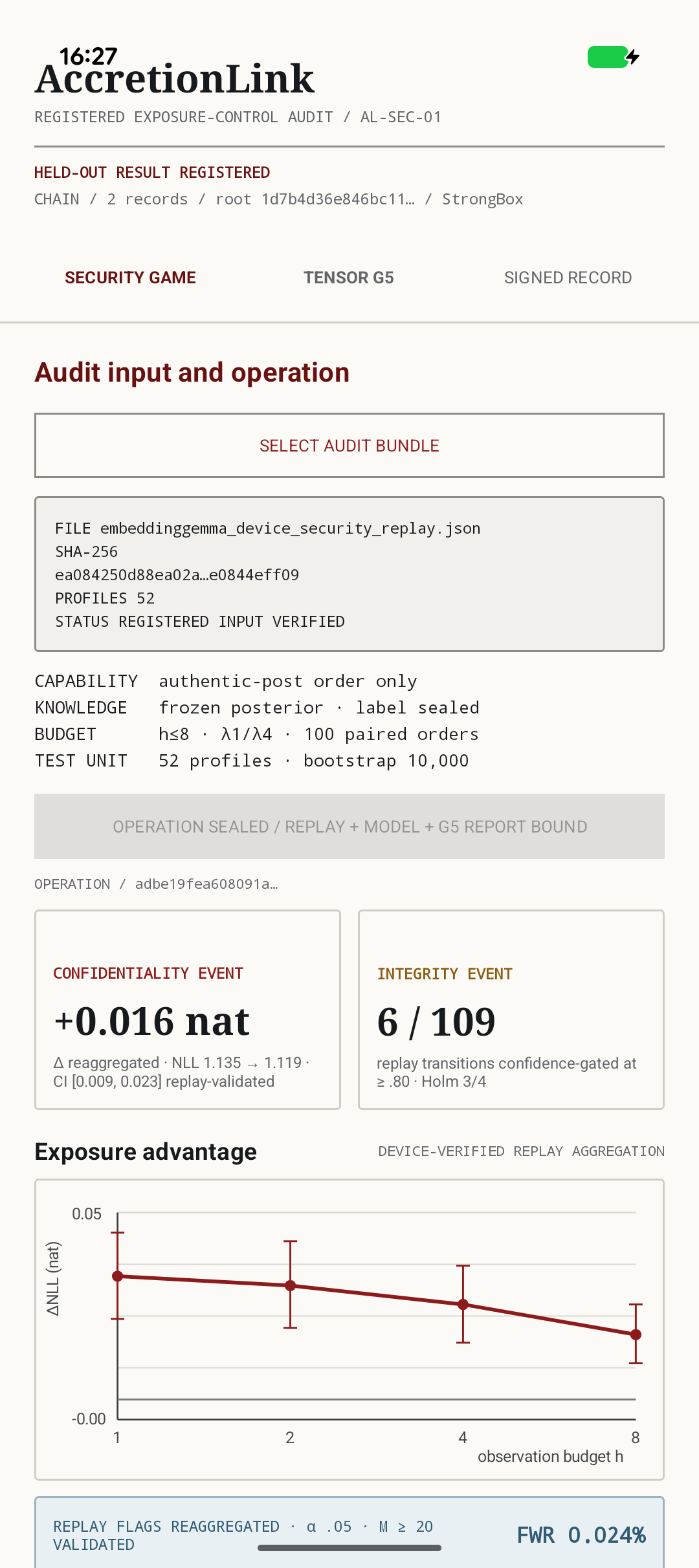}
\caption{Registered exposure-control operation on Pixel 10. The selected
bundle passed SHA-256 verification; the device reaggregated its profile effects
and event counts, validated replay-only statistics, and bound the result to the
actual model and native G5-report digests. The header shows the two-record chain
state immediately after execution. The visible operation prefix
\texttt{adbe19fe} and root prefix \texttt{1d7b4d36} match record 1 and the
authenticated prefix of the accompanying three-record checkpoint.}
\label{fig:pixel-screen}
\end{figure}
}{}
}{}

\section{Scope and limitations}
\label{sec:limitations}

\paragraph{Study populations.}
SynthPAI consists of designed, fully synthetic profiles
\citep{yukhymenko2024synthpai}. Its English text, attribute taxonomy, and
generator-specific regularities limit transfer to human writing. The
exploratory PAN15 check adds real 2015 Twitter profiles, but it has 142
official test profiles, no platform-ranking log, and benchmark-supplied age
groups whose provenance we did not independently verify. It does not justify a
claim about current Instagram users or any individual. Raw tweets, profile
identifiers, labels, embeddings, and device sidecars are excluded from the
release.

\paragraph{Selection model.}
The sensitivity bound $U$ is an input to the audit rather than an identified
property of a ranking service. Partial-identification intervals inherit that
assumption. The e-process guarantee requires the conditional cue-rate bound at
each step. If $U$ is understated, Proposition~\ref{prop:misspecified-u} gives a
finite-horizon inflation envelope, but it does not restore nominal false-warning
control; the envelope can grow exponentially with the audit horizon. Exact-text
deduplication removes literal repetitions but does not remove paraphrases or
correlated topic bursts. Posterior temperature scaling addresses marginal
calibration on the calibration split, not model misspecification.

\paragraph{Dual use.}
The confidence-greedy selector is an attribute-inference attack and can inform
both red-team audits and offensive ranking. The released experiments export no
raw third-party text and report aggregate security outcomes. These restrictions
reduce direct reuse against a named individual but do not eliminate misuse of
the selection principle.

\paragraph{Device boundary.}
Local execution removes network transmission by the audit package, while the
operating system, keyboard, display-capture path, backups, and signing-key
state remain part of the trusted device boundary. The StrongBox signature
authenticates a chain prefix under its local key. It does not attest model
correctness or reveal an omitted prefix unless an earlier checkpoint has been
fixed outside the device. The native G5 report has no remote freshness
challenge, and its execution-policy fields rely on the trusted local
measurement path. Tensor G5 measurements locate the dominant measured service
stage; they do not identify the limiting arithmetic or memory resource inside
the accelerator.

\section{Conclusion}

Exposure order is a security control surface even when every item is
authentic. In the held-out experiment, bounded selection improved the fitted
attribute adversary's log score at every prefix, three of four primary effects
survived familywise correction, and a label-blind selector created six
high-confidence false profiles. Partial identification and e-processes place
those effects inside an explicit probability model rather than treating model
confidence as security evidence.

On PAN15, the proper-score effect had the same sign at every horizon, while the
integrity attack yielded no reversal among 71 eligible profiles. A post-hoc
control retained the score advantage when a separately fitted TF--IDF model
selected inputs for the frozen G5 target, and matched identity-shuffled weights
did not reproduce the positive effect. Thus literal selector--target reuse is
not necessary for the observed PAN15 contrast. The two models still share an
age task and training corpus, and the bounded selection experiment neither
reconstructs a live recommender nor characterizes a current social-media
population.

The Pixel 10 realization uses once-only semantic extraction as the security
game's accelerator workload. Each new post crosses one fallback-free Tensor G5
AOT boundary, while a fixed state advances the posterior and sequential evidence. On the registered stream, G5
service time exceeded SentencePiece by about 22-fold and the posterior update
by more than three orders of magnitude, moving the measured bottleneck away
from the CPU host. The app then reaggregates replay effects and event counts
and commits the selected-replay, actual-model, validated-native-report, and
operation digests to a hash chain. On the measured handset, \texttt{KeyInfo}
locally reported the P-256 signing key as StrongBox-backed. This combination
turns exposure-amplified inference from a ranking intuition into a reproducible
on-device security operation.

\bibliographystyle{unsrtnat}
\bibliography{references}

\appendix
\section{Reproducibility contract}

The source compiles with \verb|pdflatex|, \verb|bibtex|, and two final
\verb|pdflatex| passes. The ancillary record contains acquisition instructions,
checksums, deterministic split identifiers, experiment configurations, model
contracts, source code, tests, and machine-readable aggregate results. Cached
SynthPAI text and private signing keys are excluded.

Each numeric result is traceable to a JSON record that binds the dataset,
configuration, split manifest, source tree, and model under SHA-256. The
held-out record also binds the pretest freeze digest and one-use test manifest.
Device records include the build fingerprint, runtime version, accelerator
logs, warm-up policy, timed samples, and thermal observations. The security
replay stores hashed profile identifiers and profile-level numerical effects
without text, labels, or embeddings.

\section{Transcript verifier and reduction}
\label{app:transcript-proof}

This appendix states the integrity guarantee of
Equations~\eqref{eq:recordchain}--\eqref{eq:checkpoint}. It does not assign
cryptographic meaning to the exposure score or establish that the model,
operating system, or displayed result is correct.

\paragraph{Verifier.}
Parse the export using the fixed field order and length-prefixed encodings
implemented by \texttt{AuditTranscript}. Each record contains input, model,
execution-evidence, and output digests in that order. Reject a malformed field,
non-finite score, nonconsecutive index, digest of length other than 32 bytes,
or wrong format domain. Starting with $\hat r_{-1}=0^{256}$, require the predecessor in
record $i$ to equal $\hat r_{i-1}$ and compute
$\hat r_i=H(\mathsf{enc}(R_i,\hat r_{i-1}))$. Require the checkpoint count to
equal $|T|$, its root to equal $\hat r_{|T|-1}$ (or $0^{256}$ for an empty
transcript), and
\[
\mathsf{Verify}_{pk^\star}
  (\mathsf{enc}_{C}(|T|,t_C,\hat r_{|T|-1}),\sigma)=1.
\]
Here $pk^\star$ is supplied by the relying party; accepting a public key merely
because it appears in the same untrusted JSON would authenticate no prior
identity. For rollback protection, the verifier also receives an externally
retained state $S=(n_0,r_0)$ from an earlier accepted checkpoint and verifies
that $T[0:n_0]$ recomputes to $r_0$.

The application's \texttt{AuditTranscript.verifyExport} checks internal
consistency using the SPKI carried in the JSON. It neither pins $pk^\star$ nor
accepts rollback state $S$. The ancillary external verifier pins the measured
SPKI fingerprint and reconstructs the same canonical bytes; rollback checking
still requires a previously retained state. Theorem~\ref{thm:transcript}
applies to $V(pk^\star,T,C,S)$, not to an unpinned self-consistency check.

\begin{theorem}[Accepted divergence]
\label{thm:transcript}
Assume the signature scheme is existentially unforgeable under adaptive
chosen-message attack and $H$ is collision resistant. Let the verifier pin
$pk^\star$ before the adversarial export and faithfully retain $S$. If an
adversary, after obtaining signatures on adaptively chosen valid checkpoint
messages, causes $V$ to accept a transcript that is not the signed history for
that checkpoint or that precedes $S$, then it produces either (i) a valid
signature on a message never submitted to the signing oracle, (ii) a collision
in $H$, or (iii) a violation of the assumed external rollback state.
\end{theorem}

\begin{proof}
Let $m=\mathsf{enc}_{C}(n,t_C,r)$ be the message of the accepted checkpoint.
If $m$ was never queried, its accepted signature is an EUF-CMA forgery in the
standard adaptive chosen-message experiment\citep{goldwasser1988signatures}.
Otherwise consider the valid history associated with the oracle query for
$m$. The count, timestamp, and root are inside the signed encoding, so an
accepted divergent transcript has the same $n$ and terminal root $r$.

Let $j$ be the first divergent record. If the two canonical byte strings at
$j$ differ but their hashes are equal, they form a collision in $H$. If their
hashes differ, acceptance forces the chains eventually to merge because their
signed terminal roots are equal. At the first merge, two distinct canonical
inputs have the same hash, again yielding a collision. Length prefixes, fixed
field order, and separate record/checkpoint domains rule out an encoding
ambiguity. Thus, absent a forgery or collision, the accepted transcript equals
the history whose checkpoint message was signed.

Finally, $V$ rejects $n<n_0$ and recomputes the externally retained root over
the first $n_0$ records. A rollback or replacement before that boundary can be
accepted only if the purported $S$ differs from the faithfully retained state,
which is case~(iii), or if the preceding collision/forgery argument applies.
\end{proof}

The theorem is conditional on these verification inputs. The current application keeps its
mutable transcript snapshot in app-private storage, holds records in memory
during a process lifetime, and exports a freshly signed checkpoint. It does
not implement a hardware monotonic counter or an off-device transparency log.
Consequently, deleting an unexported suffix or restoring both local file and
application state is outside the proved detection guarantee. An earlier
checkpoint must be retained by an independent relying party for $S$ to exist.
The local \texttt{KeyInfo} security-level query supports only an on-device
statement that the exercised key was hardware-backed. The exported label is
not covered by the checkpoint signature; Android remote key attestation
would additionally require a fresh challenge, the attestation certificate
chain, trusted-root and revocation checks, and off-device validation
\citep{androidkeyattestation}.

\section{Derivation for a general lower and upper selection bound}

If $\lambda\in[L,U]$ rather than the symmetric interval
$[1/U,U]$, inversion of Equation~\eqref{eq:selection} gives
\[
\theta(\lambda)=\frac{r}{\lambda+(1-\lambda)r}.
\]
Because
\[
\frac{\partial\theta}{\partial\lambda}
=-\frac{r(1-r)}{[\lambda+(1-\lambda)r]^2}\leq0,
\]
the lower endpoint uses $U$ and the upper endpoint uses $L$. At $r=0$ or
$r=1$, the compatible interval collapses at the corresponding boundary,
subject to finite positive $L$ and $U$.

\section{Proofs for the heterogeneous selection audit}
\label{app:selection-proofs}

This appendix gives the complete arguments behind
Theorem~\ref{thm:heterogeneous-selection} and
Propositions~\ref{prop:misspecified-u}--\ref{prop:finite-crossing}.  All
conditional claims are understood on the nondegenerate range
$0<p^*_{u,k,t}<1$; boundary cases follow directly or by limits.

\subsection{From risk-set bounds to the composite null}

For $v>0$ and $x\in(0,1)$, direct differentiation gives
\[
\frac{\partial\phi(v,x)}{\partial v}
=\frac{x(1-x)}{\{1-x+vx\}^2}\geq0,
\qquad
\frac{\partial\phi(v,x)}{\partial x}
=\frac{v}{\{1-x+vx\}^2}>0.
\]
Conditionally on $\mathcal F_{u,k,t-1}$, let
$\theta_{u,k,t}$ be the cue prevalence in the remaining source risk set and
$\lambda_{u,k,t}$ the release odds multiplier.  The one-step weighted-release
model gives
$p_{u,k,t}=\phi(\lambda_{u,k,t},\theta_{u,k,t})$. Hence
$\theta_{u,k,t}\leq\tau_{u,k,t}$ and
$\lambda_{u,k,t}\leq U_{u,k,t}$ imply
\[
p_{u,k,t}\leq
\phi(U_{u,k,t},\tau_{u,k,t})=p^*_{u,k,t}.
\]
Because the inequalities are imposed separately after every observed history,
they define a composite family of possibly nonstationary, dependent laws, not
a single Bernoulli model.

\subsection{Proof of Theorem~\ref{thm:heterogeneous-selection}}

Fix $u,k$, and $\delta>1$.  The next factor of
Equation~\eqref{eq:eprocess} has conditional expectation
\begin{align*}
\mathbb E\!\left[
\left.
\frac{\delta^{C_{u,k,t}}}{1+(\delta-1)p^*_{u,k,t}}
\right|\mathcal F_{u,k,t-1}\right]
&=\frac{1+(\delta-1)p_{u,k,t}}
        {1+(\delta-1)p^*_{u,k,t}}\\
&\leq1.
\end{align*}
The factor is nonnegative and the accumulated product starts at one, so
$E_{u,k,t}(\delta)$ is a nonnegative supermartingale.  For fixed nonnegative
weights summing to one, conditional linearity gives
\[
\mathbb E[E^{\mathrm{mix}}_{u,k,t}\mid\mathcal F_{u,k,t-1}]
\leq E^{\mathrm{mix}}_{u,k,t-1}.
\]
Ville's inequality therefore proves Equation~\eqref{eq:ville}.

Let $\mathcal G_{u-1}$ contain all completed profile blocks and the design
information used to choose $a_{u,k}$ and $\eta_u$. Terminal validity gives
$\mathbb E[e_{u,k}\mid\mathcal G_{u-1}]\leq1$ for every $k$. No joint
factorization is needed: even under arbitrary dependence among the streams,
\[
\mathbb E[B_u\mid\mathcal G_{u-1}]
=\sum_k a_{u,k}\mathbb E[e_{u,k}\mid\mathcal G_{u-1}]
\leq1.
\]
Consequently, with $Z_u=1-\eta_u+\eta_uB_u$,
$\mathbb E[Z_u\mid\mathcal G_{u-1}]\leq1$. Iterated conditioning shows that
$M_n^{\mathrm{sel}}=\prod_{u=1}^nZ_u$ is a nonnegative supermartingale;
another application of Ville's inequality gives the profile-level anytime
bound. This also shows why multiplying the dependent $e_{u,k}$ within a
profile would be unjustified.

Let $\nu_{u,k,t}$ be
$\mathcal F_{u,k,t-1}$-measurable and satisfy
$0\leq\nu_{u,k,t}\leq1/p^*_{u,k,t}$. The factor
$1+\nu_{u,k,t}(C_{u,k,t}-p^*_{u,k,t})$ is nonnegative for both cue outcomes,
and
\[
\mathbb E[1+\nu_{u,k,t}(C_{u,k,t}-p^*_{u,k,t})
\mid\mathcal F_{u,k,t-1}]
=1+\nu_{u,k,t}(p_{u,k,t}-p^*_{u,k,t})\leq1.
\]
Its product is therefore a valid predictable-betting supermartingale as
claimed.

\subsection{Proof of Proposition~\ref{prop:misspecified-u}}

Let $\widetilde E_{t,j}$ denote component $j$ computed with the assumed
ceilings $\widetilde p_i$. Its next factor $\widetilde f_{i,j}$ satisfies,
under the true bound $p_i\leq p_i^+$,
\[
\mathbb E[\widetilde f_{i,j}\mid\mathcal F_{i-1}]
=\frac{1+(\delta_j-1)p_i}{1+(\delta_j-1)\widetilde p_i}
\leq\kappa_{i,j}.
\]
Thus $S_{t,j}:=\widetilde E_{t,j}/K_{t,j}$ is a nonnegative
supermartingale, where $K_{t,j}=\prod_{i=1}^t\kappa_{i,j}$. Define
$K_t^{\max}=\prod_{i=1}^t\max_\ell\kappa_{i,\ell}$ and
$R_{t,j}=K_{t,j}/K_t^{\max}$. The predictable sequence $R_{t,j}$ is
nonincreasing because its next multiplicative factor is at most one. Hence
\[
\frac{\widetilde E_t^{\mathrm{mix}}}{K_t^{\max}}
=\sum_jw_jR_{t,j}S_{t,j}
\]
is a nonnegative supermartingale: conditioning each summand at time $t$ first
uses the supermartingale property of $S_{t,j}$ and then
$R_{t,j}\leq R_{t-1,j}$. Moreover, $p_i^+\geq\widetilde p_i$ implies
$\kappa_{i,j}\geq1$, so $K_t^{\max}\leq K_T^{\max}$ for $t\leq T$. Therefore
\begin{align*}
\left\{\sup_{t\leq T}\widetilde E_t^{\mathrm{mix}}\geq\alpha^{-1}\right\}
&\subseteq
\left\{\sup_{t\leq T}
\frac{\widetilde E_t^{\mathrm{mix}}}{K_t^{\max}}
\geq\frac{1}{\alpha K_T^{\max}}\right\}.
\end{align*}
Ville's inequality, capped by the trivial bound one, proves
Equation~\eqref{eq:misspecinflation}.

For the stationary boundary diagnostic, let the assumed rate be $a$, the
actual rate be $b>a$, and set $q=\phi(\delta,a)$. The expected one-step log
factor is
\[
g(\delta)=b\log\delta-\log\{1+(\delta-1)a\}.
\]
Since the factor is the Bernoulli likelihood ratio of $q$ against $a$,
adding and subtracting the Bernoulli log likelihood at $b$ gives
\[
g(\delta)=D_{\mathrm{KL}}(b\|a)-D_{\mathrm{KL}}(b\|q).
\]
The second term is minimized at $q=b$, equivalent to
$\delta=\operatorname{odds}(b)/\operatorname{odds}(a)$. This proves
Equation~\eqref{eq:misspecgrowth} and the stated maximum. In the constant case
the multiplicative envelope itself is
\[
K_T(\delta)=
\left\{\frac{1+(\delta-1)b}{1+(\delta-1)a}\right\}^{T},
\]
which makes the potential exponential sensitivity to an understated bound
explicit.

\subsection{Proof of Proposition~\ref{prop:finite-crossing}}

Let $S_T=\sum_{i=1}^TC_i$. For a constant null boundary and fixed alternative,
the terminal log evidence is exactly
\[
\log E_T(q)=S_TA+(T-S_T)B=TB+DS_T.
\]
Here $A>0$, $B<0$, and $D>0$. Therefore
$E_T(q)\geq1/\alpha$ if and only if $S_T\geq h_T$. A terminal crossing implies
a crossing by time $T$, while under the diagnostic model
$S_T\sim\operatorname{Bin}(T,r)$. This proves the binomial-tail lower bound in
Equation~\eqref{eq:finitepower}, including the zero/one conventions when
$h_T$ lies outside $\{0,\ldots,T\}$.

If $x_T<r$, the complement of the terminal-threshold event implies
$S_T/T<x_T$. Hoeffding's inequality for bounded independent variables gives
\[
\Pr_r(S_T/T<x_T)\leq\exp\{-2T(r-x_T)^2\},
\]
which proves the displayed analytic lower bound. Finally, the strong law of
large numbers gives $S_t/t\to r$ almost surely, so division of
$\log E_t(q)=tB+DS_t$ by $t$ yields
\[
B+Dr
=r\log\frac{q}{p_0}+(1-r)\log\frac{1-q}{1-p_0}
=D_{\mathrm{KL}}(r\|p_0)-D_{\mathrm{KL}}(r\|q).
\]
If this limit is positive, $\log E_t(q)$ diverges and the fixed threshold is
eventually crossed. For a mixture,
$E_t^{\mathrm{mix}}\geq w_jE_t(q_j)$ pointwise, proving the component-wise
extension.

\section{E-process implementation notes}

The implementation accumulates Equation~\eqref{eq:eprocess} in log space and
reports both $\log E_t$ and the threshold $\log(1/\alpha)$. A finite mixture
uses predeclared odds alternatives $\delta_j>1$ and weights $w_j\geq0$ summing
to one:
\[
E_t^{\mathrm{mix}}=\sum_jw_jE_t(\delta_j).
\]
The equivalent constant-bound parametrization is
$q_j=\phi(\delta_j,p_0)>p_0$. Components and weights are fixed before observing
the audited stream. With predictable per-step bounds, each log factor uses the
current $p^*_{u,k,t}$ from Equation~\eqref{eq:p0}; the stored bound trajectory
is part of the audit record.

\end{document}